\documentclass{article}
\usepackage{spconf,amsmath,graphicx,bm,amssymb,dsfont}
\usepackage{algorithm}
\usepackage{algorithmic}
\usepackage{cases}
\usepackage{nccmath}
\usepackage{cite}
\usepackage{subcaption}

\newtheorem{theorem}{Theorem~}
\newtheorem{lemma}{Lemma~}
\newenvironment{proof}{\qquad \emph{Proof:}~}{\hfill{$\square$}}

\newcommand{\xiaowuhao}{\fontsize{8.5pt}{\baselineskip}\selectfont}

\def\d{\ldots}
\def\cK{\mathcal{K}}
\def\cM{\mathcal{M}}

\def\cCN{\mathcal{CN}}

\def\st{\mathrm{s.t.}}

\def\rk{\mathrm{rank}}

\def\et{\eta_m}

\renewcommand{\v}[1]{\boldsymbol{#1}}
\newcommand{\m}[1]{\mathbf{#1}}

\def\hermitian{\dagger}
\def\transpose{\mathsf{T}}

\newcommand{\algref}[1]{Algorithm \ref{#1}}

\makeatletter
\newcommand{\pushright}[1]{\ifmeasuring@#1\else\omit\hfill$\displaystyle#1$\fi\ignorespaces}
\newcommand{\pushleft}[1]{\ifmeasuring@#1\else\omit$\displaystyle#1$\hfill\fi\ignorespaces}
\makeatother

\newcommand{\bk}[1][ ]{\left\{\beta_k^{#1}\right\}}
\newcommand{\Lm}[1][ ]{\left\{\m{\Lambda}_m^{#1}\right\}}
\newcommand{\Lmbk}[1][ ]{\left\{\m{\Lambda}_m\left(\left\{\beta_k^{#1}\right\}\right)\right\}}

\newcommand{\Vk}[1][ ]{\{\m{V}_k^{#1}\}}
\newcommand{\pk}[1][ ]{\{p_k^{#1}\}}

\newcommand\eqstd{\eqref{equ:dual_const2}--\eqref{equ:dual_var1}}
\newcommand\eqstp{\eqref{equ:slack_1}--\eqref{equ:pri_var2}}
\newcommand\eqsta{\eqref{equ:dual_const2}--\eqref{equ:pri_var2}}

\newcommand{\ak}{\m{V}_k \bullet \m{H}_k - \gamma_k \left(\sum_{j\neq k} \m{V}_j \bullet \m{H}_k + \m{Q}\bullet \m{H}_k + \sigma^2\right)}
\newcommand{\Bm}{\et \left[ \begin{matrix}
	\m{0}& \m{0} \\
	\m{0}& \m{Q}^{(m:M, m:M)}
\end{matrix} \right] - \m{E}_m^\hermitian \left(\sum_{k=1}^K \m{V}_k+ \m{Q}\right) \m{E}_m}

\newcommand{\G}[1][ ]{\m{I} + \sum_{m=1}^M \m{E}_m^\hermitian\m{\Lambda}_m^{#1} \m{E}_m + \sum_{j\neq k} \beta_j^{#1} \gamma_j \m{H}_j}
\newcommand{\Fm}[3][j]{\sum_{#1=#2}^M \eta_{#1} \left[ \begin{matrix}
	\m{0}& \m{0} \\
	\m{0}& \m{\Lambda}_{#1}^{(#1:M,#1:M)}
\end{matrix} \right] - \sum_{#1=#3}^M \m{E}_{#1}^\hermitian \m{\Lambda}_{#1} \m{E}_{#1}}

\def\ouralg{the proposed algorithm }

\def\fk{~\forall ~k \in \cK}
\def\fm{~\forall ~m \in \cM}
\title{Efficiently and Globally Solving Joint Beamforming and Compression Problem in the Cooperative Cellular Network via Lagrangian Duality}
\name{Xilai Fan$^{\star,\S}$, Ya-Feng Liu$^{\star}$, Liang Liu$^{\dag}$}

\address{$^{\star}$LSEC, ICMSEC, AMSS, Chinese Academy of Sciences, Beijing, China \\[2pt]
    $^{\S}$School of Mathematical Sciences, University of Chinese Academy of Sciences, Beijing, China\\[2pt]
	$^{\dag}$EIE Department, The Hong Kong Polytechnic University\\[2pt]
  Email: fanxilai21@mails.ucas.ac.cn, yafliu@lsec.cc.ac.cn, liang-eie.liu@polyu.edu.hk}
\begin{document}
\ninept
\maketitle
%
\begin{abstract}
	
Consider the joint beamforming and quantization problem in the cooperative cellular network, where multiple relay-like base stations (BSs) connected to the central processor (CP) via rate-limited fronthaul links cooperatively serve the users. This problem can be formulated as the minimization of the total transmit power, subject to all users' signal-to-interference-plus-noise-ratio (SINR) constraints and all relay-like BSs' fronthaul rate constraints.  
%
In this paper, we first show that there is no duality gap between the considered problem and its Lagrangian dual 
by showing the tightness of the semidefinite relaxation (SDR) of the considered problem. 
Then we propose an efficient algorithm based on Lagrangian duality for solving the considered problem. 
The proposed algorithm judiciously exploits the special structure of the Karush-Kuhn-Tucker (KKT) conditions of the considered problem and finds the solution that satisfies the KKT conditions via two fixed-point iterations. 
%
The proposed algorithm is highly efficient (as evaluating the functions in both fixed-point iterations are computationally cheap) 
and is guaranteed to find the global solution of the problem. 
Simulation results show the efficiency and the correctness of the proposed algorithm.    




%
%
\end{abstract}
\section{Introduction}
Lagrangian duality \cite{bertsekas_nonlinear_1999,Boyd04}, 
a principle that (convex) optimization problems can be viewed from either primal or dual perspective, 
is a powerful and vital tool in revealing the special structures of the optimization problems
arising from engineering and further better solving the problems. 
Celebrated uplink-downlink duality \cite{visotsky_optimum_1999,rashid-farrokhi_joint_1998} in wireless communications 
is an engineering interpretation of Lagrangian duality. 
Usually, the uplink problems, e.g., the transmit power minimization problems subject to quality-of-service (QoS) constraints, can be solved efficiently and globally via the fixed-point iteration algorithm. The uplink-downlink duality result thus enables efficient algorithms for solving the downlink problems 
via solving the relatively easy uplink problems. In the literature, Lagrangian duality and uplink-downlink duality results have been proved in various ways and applied to solve different downlink problems; see \cite{visotsky_optimum_1999,rashid-farrokhi_joint_1998,
	duality11,schubert_iterative_2005,rashid-farrokhi_transmit_1998,
	bengtsson_optimal_1999,bengtsson_optimum_2002,wiesel_linear_2006,
	dahrouj_coordinated_2010,duality10,duality7,duality9,duality3, duality4,Tsung-Hui18,Yafeng13,Liu15,Liang15} and the references therein. 

Different from the above works, this paper considers the joint beamforming and quantization problem in the cooperative cellular network, 
where multiple relay-like base stations (BSs) are connected to the central processor (CP) via rate-limited fronthaul links to cooperatively serve the users for 
effectively mitigating multiuser intercell interference. 
Such network includes coordinated multipoint \cite{Irmer11}, 
cloud radio access network \cite{Simeone16}, and cell-free massive multi-input multi-output \cite{Ngo17} as special cases. 
Recently, Refs. \cite{Liang16,liu_uplink-downlink_2020} have established an interesting uplink-downlink duality for such network when relay-like BS compression optimization is considered. Specifically, given the same beamforming vectors in the uplink and downlink, it has been shown there that when Wyner-Ziv compression and multivariate compression are adopted in the uplink and the downlink, respectively, the transmit power minimization problem in the uplink subject to individual signal-to-interference-plus-noise-ratio (SINR) constraints and fronthaul capacity constraints is equivalent to that in the downlink. 
Furthermore, \cite{liu_uplink-downlink_2020} has designed an algorithm for solving the joint beamforming and quantization problem 
based on the established duality result. 
The algorithm in \cite{liu_uplink-downlink_2020} first solves the uplink problem 
via fixed-point iteration and then solves the downlink problem with fixed beamformers 
(which is a convex problem) obtained by solving the uplink problem by calling a solver. 

%
%
%
%


In this paper, we consider the same joint beamforming and quantization problem as in \cite{liu_uplink-downlink_2020} 
but make further progress in the duality result and the algorithm. 
The main contributions of this paper are as follows. (1) \emph{New Duality Result.} We establish the tightness of the semidefinite relaxation (SDR) of the considered problem 
	and thus the equivalence of the two problems. 
	This result further implies that the dual problems of the considered problem and its SDR are the same.
	Note that the Lagrangian dual of the original problem is studied in this paper, 
	which differs from the Lagrangian dual of the problem with fixed beamformers in \cite{liu_uplink-downlink_2020}. 
(2) \emph{Efficient Fixed-Point Iteration Algorithm.} Based on the established duality result, 
	we propose an efficient algorithm for solving the considered problem. 
	The proposed algorithm first solves the dual problem via fixed-point iteration 
	and then solves the primal problem via another fixed-point iteration. 
	The proposed algorithm is highly efficient (as each update of variables in fixed-point iterations is computationally cheap) 
	and is guaranteed to find the global solution of the problem. 
	The proposed algorithm exploits more special structures of the solution of the considered problem 
	than the algorithm in \cite{liu_uplink-downlink_2020} 
	and thus significantly outperforms it in terms of the computational efficiency. 
\emph{Notations.} For any matrix $\m{A}$, $\m{A}^\hermitian$ and $\m{A}^\transpose$ denote the conjugate transpose and transpose of $\m{A}$, respectively; 
$\m{A}^{(m, n)}$ denotes the entry on the $m$-th row and the $n$-th column of $\m{A};$ and
$\m{A}^{(m_1:m_2, n_1:n_2)}$ denotes a submatrix of $\m{A}$ defined by 
\begin{equation*}
	\left[ \begin{matrix}
		\m{A}^{(m_1, n_1)}& \cdots & \m{A}^{(m_1, n_2)} \\
		\vdots& \ddots & \vdots\\
		\m{A}^{(m_2, n_1)}& \cdots & \m{A}^{(m_2, n_2)}
	\end{matrix} \right].
\end{equation*}
For two matrices $\m{A}_1$ and $\m{A}_2$ of appropriate sizes
, $\m{A}_1 \bullet \m{A}_2$ denotes the trace of $\m{A}_1 \m{A}_2$. 
We use $\cCN(\m{0}, \m{Q})$ to denote the complex Gaussian distribution with zero mean and covariance $\m{Q}.$ ~
Finally, we use $\m{I}$ to denote the identity matrix of an appropriate size, 
$\m{0}$ to denote an all-zero matrix of an appropriate size, and $\m{E}_m$ to denote the square all-zero matrix except its $m$-th diagonal entry being one. 
		
\section{System Model and Problem Formulation}
\subsection{System Model}
Consider a cooperative cellular network consisting of one CP and $M$ single-antenna relay-like BSs (will be called relays for short later), 
which cooperatively serve $K$ single-antenna users. 
 In such network, all users and relays are connected by noisy wireless channels and all relays and the CP are connected by noiseless fronhaul links of finite rate. Let $\cM$ and $\cK$ denote the sets of the relays and the users, respectively. 

We first introduce the compression model from the CP to the relays. The transmitted signal at the CP is $\tilde{\v{x}} = \sum_{k=1}^K \v{v}_k s_k$, 
where $\v{v}_k = [v_{k,1}, \d , v_{k,M}]^\transpose$ is an $M \times 1$ beamforming vector and $s_k\sim \cCN(0, 1)$ is the information signal for user $k$. 
Because of the limited capacities of the fronthaul links, 
the signal from the CP to the relays need to be first compressed before transmitted. Let the compression error be $\v{e} = [e_1, \d, e_m]^\transpose \sim \cCN(\m{0}, \m{Q}),$ where $\m{Q}$ is the covariance matrix to be designed.  
Then the received signal at relay $m$ is $x_m = \sum_{k=1}^K v_{k,m}s_k + e_m$. 
The channel model from the relays to the users is 
$y_k = \sum_{m=1}^M h_{k,m} x_m + z_k,$
where $y_k$ is the signal received by user $k$, $x_m$ is the signal transmitted by relay $m$, 
$h_{k,m}$ is the channel coefficient from relay $m$ to user $k$, and 
$\left\{z_k\right\}$ are independent and identically distributed (i.i.d.) additive Gaussian noise distributed as $\cCN(0, \sigma^2).$
 
Under the above model, the received signal at user $k$ is 
\begin{equation}
	y_k = \v{h}_k^\hermitian \left( \sum_{i=1}^K \v{v}_i s_i \right) + \v{h}_k^\hermitian \v{e} + z_k,
	\label{equ:DBM}
\end{equation}
where $\v{h}_k = [h_{k,1},\d,h_{k,M}]^\hermitian$ is the channel vector of user $k$. Then the total transmit power of all the relays is $
	\sum_{k=1}^K \|\v{v}_k\|^2 + \m{Q} \bullet \m{I};$
	the SINR of user $k$ is
\begin{equation}
	\frac{|\v{h}_k^\hermitian \v{v}_k|^2}{\sum_{j\neq k} |\v{h}_k^\hermitian \v{v}_j|^2 + \v{h}_k^\hermitian \m{Q} \v{h}_k + \sigma^2},~\forall~k\in\cK; 
\end{equation}
and the compression rate of relay $m$ under the multivariate compression strategy \cite{Simeone13} is 
\begin{equation}
	\begin{aligned}
		\log_2 \frac{\sum_{k=1}^K |v_{k,m}|^2 + \m{Q}^{(m,m)}}{ \m{Q}^{(m:M, m:M)}/\m{Q}^{(m+1:M, m+1:M)} },~\forall~m\in\cM.\\
	\end{aligned}
\end{equation}
In the above, $\m{Q}^{(m:M, m:M)}/\m{Q}^{(m+1:M, m+1:M)}$ is the Schur complement of the block $\m{Q}^{(m+1:M, m+1:M)}$ of $\m{Q}^{(m:M, m:M)},$ which is equal to $\m{Q}^{(m,m)} - \m{Q}^{(m,m+1:M)} (\m{Q}^{(m+1:M, m+1:M)})^{-1} \m{Q}^{(m+1:M, m)}.$

\subsection{Problem Formulation}
Now we are ready to present the problem formulation. Given a set of SINR targets $\left\{ \gamma_k \right\}$ and a set of fronthaul capacities $\left\{ C_m \right\}$, 
the interested optimal joint beamforming and compression problem, which minimizes the total transmit power subject to all users' SINR constraints and
all relays' fronthaul capacity constraints, is as follows\cite{liu_uplink-downlink_2020}: 
\begin{equation}
	\xiaowuhao{
		\begin{aligned}
			\min_{\{\v{v}_k\}, \m{Q}} &\quad \sum_{k=1}^K \|\v{v}_k\|^2 + \m{Q} \bullet \m{I}\\
			\st~~ &\quad \frac{|\v{h}_k^\hermitian \v{v}_k|^2}{\sum_{j\neq k} |\v{h}_k^\hermitian \v{v}_j|^2 + \v{h}_k^\hermitian \m{Q} \v{h}_k + \sigma^2} \geq \gamma_k, \fk{},\\
			&\quad \log_2 \frac{\sum_{k=1}^K |v_{k,m}|^2 + \m{Q}^{(m,m)}}{\m{Q}^{(m:M, m:M)}/\m{Q}^{(m+1:M, m+1:M)}} \leq C_m, \fm{},\\
			&\quad \m{Q}\succeq \m{0}. 
		\end{aligned}
		\label{equ:original obp}
	}
\end{equation}
Let $\m{H}_k = \v{h}_k \v{h}_k^\hermitian$ for all $k$ and $\eta_m = 2^{C_m}$ for all $m.$ By \cite[Propostion 4]{liu_uplink-downlink_2020}, problem \eqref{equ:original obp} is equivalent to the following problem: 
\begin{equation}
	\medmath{
	\begin{aligned}
		\min_{\{\v{v}_k\}, \m{Q}}&\quad \sum_{k=1}^K \|\v{v}_k\|^2 + \m{Q} \bullet \m{I}\\
		\st \quad &\quad  \v{v}_k^\hermitian \m{H}_k \v{v}_k - \gamma_k \left(\sum_{j\neq k} \v{v}_j^\hermitian \m{H}_k \v{v}_j + \m{Q}\bullet \m{H}_k + \sigma^2\right) \geq 0, \fk{}, \\
		&\quad \eta_m \left[ \begin{matrix}
			\m{0}& \m{0} \\
			\m{0}& \m{Q}^{(m:M, m:M)}
		\end{matrix} \right] - \m{E}_m^\hermitian \left(\sum_{k=1}^K \v{v}_k \v{v}_k^\hermitian+ \m{Q}\right) \m{E}_m \succeq \m{0} , \\ & \pushright{\fm{}},\\
		&\quad \m{Q}\succeq \m{0}. 
	\end{aligned} \tag{P}
	}
\end{equation}

In the following section, we shall focus on problem (P) and design an efficient algorithm for solving it. 

\section{Proposed Algorithm via Lagrangian Duality}
\subsection{Tightness of SDR of (P)}
Problem (P) is a separable homogeneous quadratically constrained quadratic program. 
A well-known technique to tackle such problem is the SDR \cite{luo_semidefinite_2010}. 
Applying the SDR technique to (P), 
we obtain 
\begin{equation}\label{SDR} 
	{
		\begin{aligned}
			\min_{\Vk, \m{Q}} &\quad \sum_{k=1}^K \m{V}_k \bullet \m{I} + \m{Q} \bullet \m{I}\\
			\st~~~ &\quad a_k(\left\{\m{V}_k\right\},\m{Q}) \geq 0,~{\fk{}}, \\
			&\quad \m{B}_m(\left\{\m{V}_k\right\},\m{Q}) \succeq \m{0},~{\fm{}},\\
			&\quad \m{V}_k \succeq \m{0} ,\fk{}, \\
			&\quad \m{Q} \succeq \m{0}, 
		\end{aligned} 
	}
\end{equation}
where
\begin{equation*}
\begin{aligned}
a_k(\left\{\m{V}_k\right\},\m{Q}) & = \ak, \\ 
\m{B}_m(\left\{\m{V}_k\right\},\m{Q}) & = \Bm.
\end{aligned}
\end{equation*} The Lagrangian dual of problem \eqref{SDR} is 
	\begin{equation}\label{dual}
		{
			\begin{aligned}
				\max_{\bk, \Lm} &\quad \sum_{k=1}^K (\gamma_k \sigma^2) \beta_k \\
				\st~~~~~~ &\quad \m{C}_k(\left\{\beta_k\right\},\left\{\m{\Lambda}_m\right\}) - \beta_k \m{H}_k \succeq \m{0},~{\fk{}},\\
				&\quad	 \m{D}(\left\{\beta_k\right\},\left\{\m{\Lambda}_m\right\}) \succeq \m{0},\\
				&\quad \beta_k \geq 0 ,\fk{},\\ &\quad \m{\Lambda}_m \succeq \m{0} ,\fm{},
			\end{aligned}
		}
	\end{equation}
where $\beta_k$ is the dual variable associated with the $k$-th SINR constraint in \eqref{SDR},
$\m{\Lambda}_m$ is the dual variable associated with the $m$-th fronthaul capacity constraint in \eqref{SDR}, and
	\begin{equation*}
\m{C}_k(\left\{\beta_k\right\},\left\{\m{\Lambda}_m\right\}) = \G, 
\label{equ:def_G}
\end{equation*} 
\begin{multline*}
\m{D}(\left\{\beta_k\right\},\left\{\m{\Lambda}_m\right\})=\m{I} - \sum_{m=1}^M \eta_m \left[ \begin{matrix}
\m{0}& \m{0} \\
\m{0}& \m{\Lambda}_m^{(m:M,m:M)}
\end{matrix} \right] \\+ \sum_{k=1}^K \beta_k \gamma_k \m{H}_k + \sum_{m=1}^M \m{E}_m^\hermitian \m{\Lambda}_m \m{E}_m.
\end{multline*}
	
An important line of research on the SDR is to study its tightness \cite{luo_semidefinite_2010,lu2019tightness,liu2021cram}. 
In the following theorem, we show that the SDR in \eqref{SDR} is tight (if it is feasible), i.e., it always has a rank-one solution. 
This shows that problem (P) admits a convex reformulation and answers a question in \cite[Section IX-B]{liu_uplink-downlink_2020}. 

\begin{theorem}
	Suppose that problem \eqref{SDR} is feasible. Then it always has a rank-one solution. 
\end{theorem}
\begin{proof}
	Since the SDR is feasible, there must exist a primal-dual pair
	$\Vk, \m{Q}, \bk, \Lm$ such that the Karush-Kuhn-Tucker (KKT) conditions of problem \eqref{SDR} hold. 
	In particular, the complimentary slackness condition $\m{V}_k \bullet (\m{C}_k(\left\{\beta_k\right\},\left\{\m{\Lambda}_m\right\}) - \beta_k \m{H}_k)= 0$ holds. 
	Since $\m{C}_k(\left\{\beta_k\right\},\left\{\m{\Lambda}_m\right\})$ is positive definite and $\m{H}_k$ is rank-one and positive semidefinite, it follows that 
	$$\rk\left(\m{C}_k(\left\{\beta_k\right\},\left\{\m{\Lambda}_m\right\}) - \beta_k \m{H}_k \right) \geq M-1,$$ 
	which, together with the complimentary slackness condition and the rank inequality, implies that 
	$\rk(\m{V}_k)\leq 1$.
\end{proof}

\subsection{Proposed Algorithm}
It is well known that the KKT conditions are sufficient and necessary for the global solution of problem \eqref{SDR}. By further exploiting the special strucrure of problem \eqref{SDR} and its KKT conditions, we get the following conditions that the solution of problem \eqref{SDR} must satisfy:  
\begin{numcases}{}
	\m{D}(\left\{\beta_k\right\},\left\{\m{\Lambda}_m\right\})=\m{0}, \label{equ:dual_const2}\\
	\left.\begin{aligned}
		&\rk(\m{\Lambda}_m)=1,~\m{\Lambda}_m \succeq \m{0},\fm{}, \\ &\m{\Lambda}_m^{(1:m-1,1:m)}= \m{0},~\m{\Lambda}_m^{(m:M,1:m-1)} = \m{0},\fm{},
	\end{aligned}\right\} \label{equ:dual_var2}\\\xiaowuhao{\left.\begin{aligned}&
	\rk(\m{C}_k(\left\{\beta_k\right\},\left\{\m{\Lambda}_m\right\}) - \beta_k \m{H}_k)=M-1, \fm{},\\
	&\m{C}_k(\left\{\beta_k\right\},\left\{\m{\Lambda}_m\right\}) - \beta_k \m{H}_k \succeq \m{0}, \fm{},\end{aligned}\right\}}\label{equ:dual_const1}\\
	\beta_k\geq 0, \fk{}, \label{equ:dual_var1}\\
	\m{V}_k\bullet\left( \m{C}_k(\left\{\beta_k\right\},\left\{\m{\Lambda}_m\right\}) - \beta_k \m{H}_k\right) =0,\fk{},\label{equ:slack_1}\\
	\m{V}_k\succeq \m{0},~\rk(\m{V}_k) = 1, \fk{},\label{equ:pri_var1}\\
	a_k(\left\{\m{V}_k\right\},\m{Q})=0,\fk{},\label{equ:pri_const1}\\
	\m{B}_m(\left\{\m{V}_k\right\},\m{Q})\succeq \m{0}, \fm{},\label{equ:pri_const2}\\
	\m{\Lambda}_m \bullet \m{B}_m(\left\{\m{V}_k\right\},\m{Q}) = 0, \fm{},\label{equ:slack_2}\\
	\m{Q}\succeq \m{0}.\label{equ:pri_var2}
\end{numcases}
The above conditions are essentially the KKT conditions of problem \eqref{SDR} except the one in \eqref{equ:dual_const2}, whose proof needs a judicious treatment of the special structure and the KKT conditions of problem \eqref{SDR}.
%

Next, we shall design an algorithm for solving the above equations by further carefully exploiting the special structures in the equtions. 
The idea is to first find $\bk$ and $\Lm$ by solving Eqs.~\eqstd{} and 
then plug $\bk$ and $\Lm$ into Eqs.~\eqstp{} and solve for $\Vk$ and $\m{Q}$. 

\subsubsection{Solving Eqs.~\eqstd{}}
 Suppose that $\bk$ are given, we first find $\Lm$ that satisfy Eqs.~\eqref{equ:dual_const2} and \eqref{equ:dual_var2}.
Define $\m{\Gamma} = \m{I} + \sum_{k=1}^K \beta_k \gamma_k \m{H}_k.$ Then
Eq. \eqref{equ:dual_const2} is equivalent to
\begin{equation}
	\Fm[m]{1}{1} = \m{\Gamma}. 
	\label{equ:D_changed}
\end{equation}
We know from the special properties of $\Lm$ in \eqref{equ:dual_var2} that only $\m{\Lambda}_1$ affects the first row and column of matrix $\m{\Gamma}.$  
Therefore, the entries in the first row of $\m{\Lambda}_1$ should be 
	$\left[ \begin{matrix}
		\frac{1}{\eta_1 - 1}\m{\Gamma}^{(1,1)}, \frac{1}{\eta_1}\m{\Gamma}^{(1,2:M)}
	\end{matrix} \right].$
Since $\m{\Lambda}_1$ is of rank one, we can further obtain all entries of $\m{\Lambda}_1$ based on its entries in the first row. 
After $\m{\Lambda}_1$ is obtained, we can subtract all terms related to $\m{\Lambda}_1$ from both sides of \eqref{equ:D_changed} and the left-hand side of \eqref{equ:D_changed} becomes
\begin{equation}
	\Fm[m]{2}{2}.
\end{equation}
Then we can do the same to find $\m{\Lambda}_2$. We repeat the above procedure until all $\Lm$ are obtained. 
It can be shown that $\Lm$ that satisfy Eqs.~\eqref{equ:dual_const2}--\eqref{equ:dual_var2} are unique, 
and such solution, depending on the given $\bk,$ is denoted as $\Lmbk$. 

To ease the presentation, define $\m{C}_k\triangleq\m{C}_k(\left\{\beta_k\right\},\left\{\m{\Lambda}_m\right\}).$
~Since $\m{C}_k \succ \m{0}$ and $\m{H}_k\succeq \m{0}$ is of rank one,  
there exists a unique $\beta_k$ such that one and only one eigenvalue of $\m{C}_k-\beta_k\m{H}_k$ is equal to zero.  
Such $\beta_k$ admits the following closed-form solution:
\begin{equation*}
	\beta_k\left(\Lm, \left\{\beta_j\right\}_{j\neq k}\right) = \left( \v{h}_k^\hermitian \m{C}_k^{-1} \m{h}_k\right)^{-1}.
	\label{equ:L2b}
\end{equation*}

From the above discussion, we know: if $\Lm$ are known, one can get $\left\{\beta_k\left(\Lm, \left\{\beta_j\right\}_{j\neq k}\right)\right\}$ 
such that \eqref{equ:dual_const1} and \eqref{equ:dual_var1} hold; 
if $\bk$ are known, one can get $\Lmbk$ such that \eqref{equ:dual_const2} and \eqref{equ:dual_var2} hold. 
If one can find $\bk$ and $\Lmbk$ that satisfy
\begin{equation}
	\beta_k = I_k\left(\bk\right) \triangleq \beta_k\left(\Lmbk, \left\{\beta_j\right\}_{j\neq k}\right),~\fk{}, 
	\label{equ:raw_dual_iter}
\end{equation}
then all Eqs.~\eqstd{} are satisfied. Define
$\v{\beta} = [\beta_1, \d, \beta_K]^\transpose$ and $I(\v{\beta}) = [I_1(\bk),\d,I_K(\bk)]^\transpose,$ 
then \eqref{equ:raw_dual_iter} becomes 
\begin{equation}
	\v{\beta} = I(\v{\beta}). 
	\label{equ:dual_fix_point}
\end{equation}  It is worth highlighting that the computational cost of evaluating the function $I(\v{\beta})$ at any given positive $\v{\beta}$ is quite cheap.

\begin{lemma}
	The function $I(\cdot)$ defined in \eqref{equ:dual_fix_point} is a standard interference function. 
	\label{lemma:dual}
\end{lemma}
From Lemma \ref{lemma:dual} and \cite[Theorem 2]{Yates95}, it follows that the following fixed-point iteration
$\v{\beta}^{(i+1)} = I(\v{\beta}^{(i)})$
will converge to the unique solution of \eqref{equ:dual_fix_point}.
Therefore, the above fixed-point iteration provides an efficient way of solving Eqs.~\eqstd{}. 

\subsubsection{Solving Eqs.~\eqstp{}}
Suppose we already have $\bk$ and $\Lm$ that satisfy Eqs.~\eqstd{}. 
We still need to find $\Vk$ and $\m{Q}$ that satisfy Eqs. \eqstp{}. By Eq. \eqref{equ:pri_var1}, let $\m{V}_k = p_k \v{v}_k \v{v}_k^\hermitian$ with $\|\v{v}_k\|=1$. 
Then Eq. \eqref{equ:slack_1} becomes
$\v{v}_k^\hermitian \left(\m{C}_k- \beta_k \m{H}_k\right) \v{v}_k = 0.$ 
Combining this and Eq. \eqref{equ:dual_const1} gives
$\left(\m{C}_k- \beta_k \m{H}_k\right) \v{v}_k=\m{C}_k \v{v}_k - \beta_k \m{h}_k \left(\m{h}_k^\hermitian \v{v}_k\right) = \m{0}.$
Hence, $\v{v}_k$ can be solved explicitly as follows: 
\begin{equation}
	\v{v}_k = \frac{\m{C}_k^{-1} \v{h}_k}{\left\|\m{C}_k^{-1} \v{h}_k\right\|}.
\label{equ:pri_beamforming_direction}
\end{equation}
Let $\m{U}_k = \v{v}_k \v{v}_k^\hermitian$. Substituting \eqref{equ:pri_beamforming_direction} into \eqref{equ:pri_const1}, one has  
\begin{equation*}
	p_k \m{U}_k \bullet \m{H}_k - \gamma_k \left(\sum_{j\neq k} p_j \m{U}_j \bullet \m{H}_k + \m{Q}\bullet \m{H}_k + \sigma^2\right) = 0.
\end{equation*}
Then one can solve for $p_k$ as follows:
\begin{equation}
	p_k\left(\m{Q}, \left\{p_j\right\}_{j\neq k}\right) = \frac{\gamma_k \left(\sum_{j\neq k} p_j \m{U}_j \bullet \m{H}_k + \m{Q}\bullet \m{H}_k + \sigma^2\right)}{\m{U}_k \bullet \m{H}_k}. 
\label{equ:Q2p}
\end{equation}

Now suppose $\pk$ are known. By Eq. \eqref{equ:dual_var2}, one can decompose $\m{\Lambda}_m$ into
$\m{\Lambda}_m = \v{\lambda}_m \v{\lambda}_m^\hermitian,$
where $\v{\lambda}_m = \left[0, \d, 0, \v{\lambda}_m^{(m)}, \d, \v{\lambda}_m^{(M)}\right]^\transpose$. 
This decomposition, together with Eqs. \eqref{equ:pri_const2} and \eqref{equ:slack_2}, implies
\begin{equation}
	\m{B}_m \v{\lambda}_m = \m{0}, \fm{}.
	\label{equ:equ_Bm}
\end{equation}
Next we solve \eqref{equ:equ_Bm} from $m=M$ to $m=1$ and can obtain the desired $\m{Q}$. 
More specifically, when $m = M$, since $\v{\lambda}_M^{(M)} > 0$, it follows that
\begin{equation}
	\m{Q}^{(M,M)} = \frac{\sum_{k=1}^K \m{V}_k^{(M,M)}}{\et - 1}. 
	\label{equ:QMM}
\end{equation}
When $m = M-1$, we can substitute \eqref{equ:QMM} into \eqref{equ:equ_Bm} to solve for 
$\m{Q}^{(M-1,M-1)}, \m{Q}^{(M,M-1)},$ and $\m{Q}^{(M-1,M)}$. 
In particular, we can obtain $\m{Q}^{(M-1,M)}$ by using the last equation of \eqref{equ:equ_Bm} with $m=M-1;$ then we can further obtain $\m{Q}^{(M-1,M-1)}$ by using the second last equation of \eqref{equ:equ_Bm} with $m=M-1$. 
In fact, each step of the above procedure admits a closed-form solution. 
We can do the same sequentially to solve problem \eqref{equ:equ_Bm} with $m=M-2, M-3, \d,1$. Denote the solution to \eqref{equ:equ_Bm} as $\m{Q}(\pk)$.

Similarly, we can define a fixed-point iteration to solve Eqs.~\eqstp{} for the desired $\Vk$ and $\m{Q}$ and show that the fixed-point iteration is a standard interference function and thus converges to the unique solution. We omit the details due to the space reason.

\subsubsection{Proposed Fixed-Point Iteration Algorithm}
Now, we present the algorithm for solving problem \eqref{SDR} (equivalent to problem (P)). 
The algorithm first finds $\bk$ and $\Lm$ that satisfy Eqs.~\eqstd{}; 
with found $\bk$ and $\Lm$ fixed, the algorithm then finds $\Vk$ and $\m{Q}$ that satisfy Eqs.~\eqstp{}. 
Hence, $\Vk$, $\m{Q}$, $\bk,$ and $\Lm$ together satisfy Eqs.~\eqsta{} and thus is a KKT point of problem \eqref{SDR}. 
Since $\rk\left(\m{V}_k\right) = 1$ for all $k$, we can recover the optimal solution for problem (P). 
The pseudocodes of the proposed algorithm are given in \algref{alg:the_alg}. 

\begin{algorithm}[H]
	\caption{Proposed Algorithm for Solving Problem (P)}
	\begin{algorithmic}[1]
		\STATE Find $\bk$ and $\Lm$ that satisfy Eqs. \eqstd{} by performing the fixed-point iteration in \eqref{equ:dual_fix_point} on $\bk$ until the desired error bound is met. 
		\STATE Find $\Vk$ and $\m{Q}$ that satisfy Eqs. \eqstp{} by performing an appropriate fixed-point iteration on $\pk$ until the desired error bound is met. 
		\STATE Find $\v{v}_k$ such that $\m{V}_k = \v{v}_k \v{v}_k^\hermitian, \fk{}$. 
		\STATE \textbf{Output:} $\{\v{v}_k\}$~and~$\m{Q}.$
	\end{algorithmic}
	\label{alg:the_alg}
\end{algorithm}

\begin{theorem}
	If the SDR in \eqref{SDR} is feasible, then proposed \algref{alg:the_alg} returns the optimal solution of problem (P). 
\end{theorem}

In addition to the global optimality, proposed Algorithm 1 is also computationally efficient, as the evaluation of the functions in fixed-point iterations is computationally cheap. 


\section{Simulation Results}



\begin{figure}[t]
\captionsetup{skip=3pt}
\begin{subfigure}{0.24\textwidth}
    \includegraphics[width=\textwidth,trim={0.4cm 5cm 1.5cm 6cm},clip]{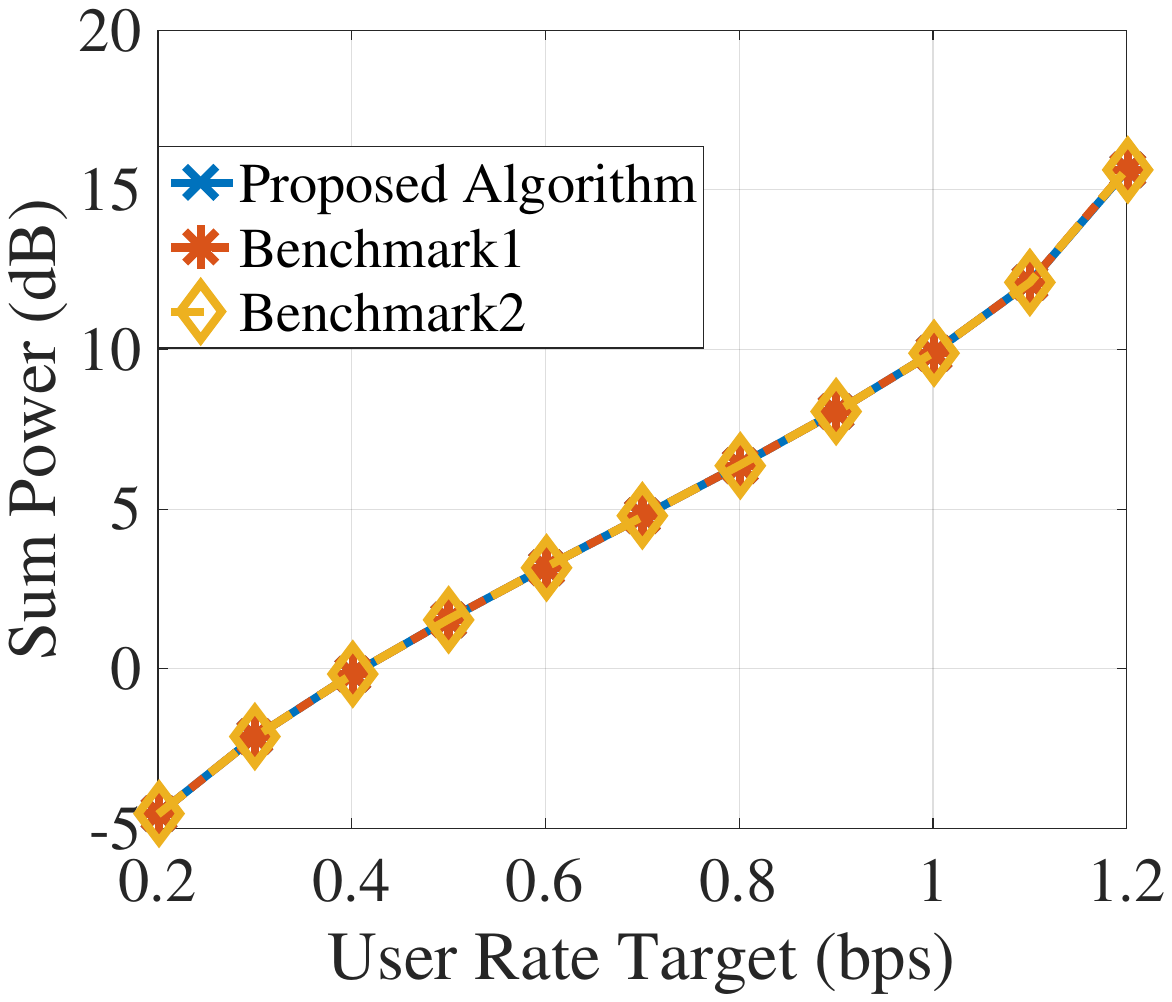}
	\vspace{-15pt}
	\caption{}
    \label{fig:opt_value_comparison}
\end{subfigure}
\hfill
\begin{subfigure}{0.24\textwidth}
    \includegraphics[width=\textwidth,trim={0.4cm 5cm 1.5cm 6cm},clip]{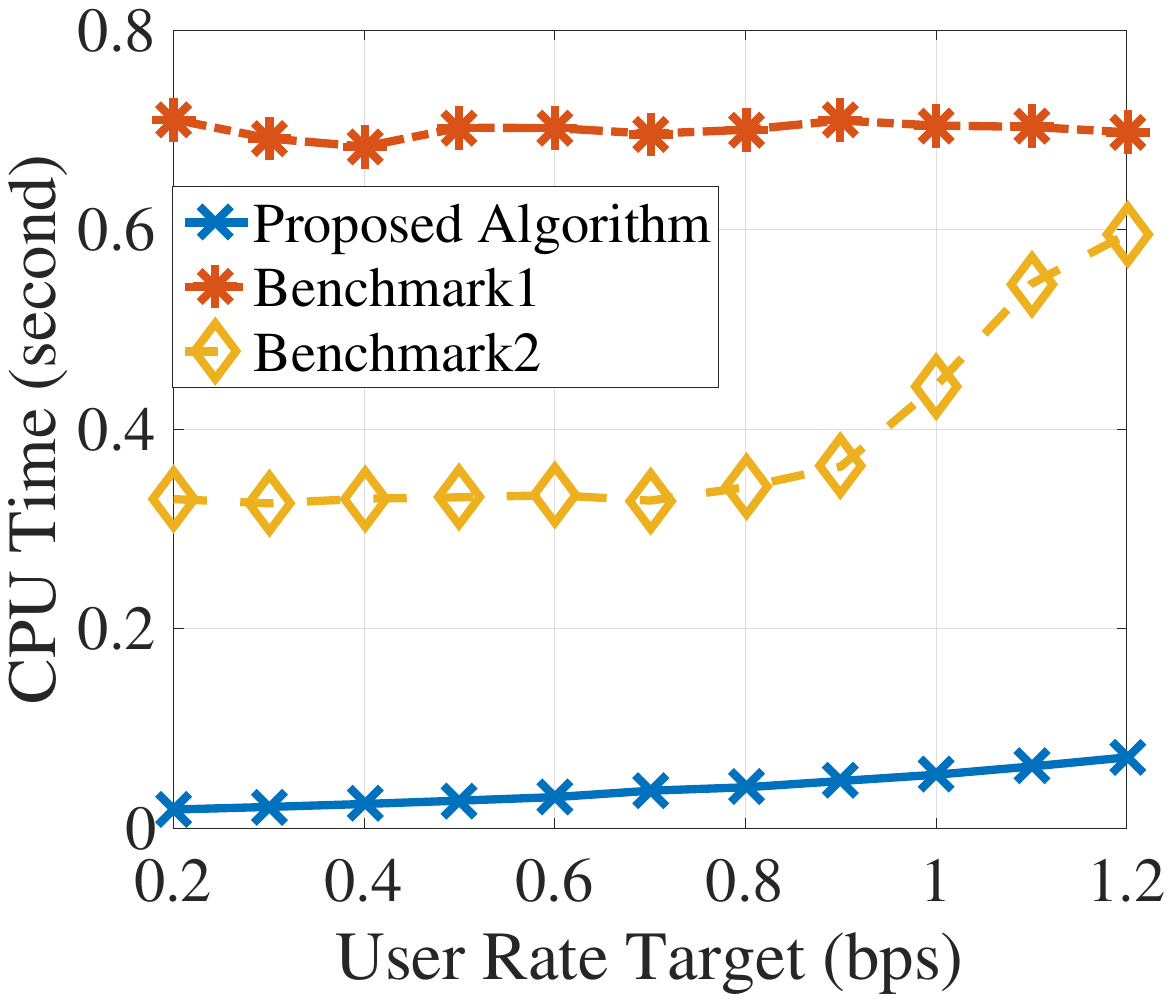}
	\vspace{-15pt}
	\caption{}
    \label{fig:time_comparison}
\end{subfigure}   
\caption{(a) Average sum power versus the user rate target; (b) Average CPU time versus the user rate target.}
\label{fig:figures}
\vspace{-10pt}
\end{figure}

In this section, we present some simulation results to show the correctness and the efficiency of proposed \algref{alg:the_alg} for solving problem (P). We consider a network with $M = 8$ relays and $K = 10$ users, 
where the wireless channels between these relays and users are generated based on the i.i.d. Rayleigh fading model following $\cCN(0, 1),$ 
and the fronthaul capacities between all relays and the CP are set to be 3 bits per symbol (bps). 
Moreover, the noise powers at the users are set to be $\sigma^2 = 1$. 
The rate targets for all the users are assumed to be identical. 
All simulation results are obtained by averaging over 200 Monte-Carlo runs. 

We compare \ouralg with the following two benchmarks. (i)  Benchmark1 is to directly call CVX \cite{CVX} to solve the SDR in \eqref{SDR}. 
This benchmark is helpful in verifying the tightness of the corresponding SDR as well as the correctness of Theorem 2. (ii) Benchmark2 is the proposed algorithm in \cite{liu_uplink-downlink_2020}. 
This algorithm first uses the fixed-point iteration to solve the dual uplink problem and 
then calls CVX to solve the reduced primal downlink problem with fixed beamforming vectors (which is a convex problem). 
The key difference of \ouralg and the algorithm in \cite{liu_uplink-downlink_2020} is that 
our proposed algorithm uses the fixed-point iteration algorithm 
for solving the primal problem by exploiting more special 
structures in the problem (after solving the dual problem).

\captionsetup{subrefformat=parens}
Fig.~\ref{fig:figures}~\subref{fig:opt_value_comparison} plots the average sum power obtained by the proposed algorithm and two benchmarks, where the user rate target ranges from 0.2 to 1.2 bps. 
We can see from the figure that all the three algorithms return the same solution. 
This verifies the tightness of the SDR (i.e., Theorem 1) and the global optimality of the solution returned by the proposed algorithm (i.e., Theorem 2). 
Fig.~\ref{fig:figures}~\subref{fig:time_comparison} plots the average CPU time taken by different algorithms. 
We can observe from the figure that: 
Benchmark2 performs much efficient than Benchmark1 
and our proposed algorithm performs the most efficient. 
In particular, our proposed algorithm significantly outperforms Benchmark2 
in terms of the CPU time and the CPU time gap tends to become large as the user rate target becomes large 
(i.e., as the corresponding problem becomes difficult). 
These preliminary results illustrate high efficiency and global optimality of our proposed algorithm 
as well as the important role of Lagrangian duality in exploiting the special problem structure in the algorithmic design. 

\newpage
\bibliographystyle{IEEEtran}

\bibliography{icassp_duality}

\newpage

\end{document}